\documentclass[conference,a4paper]{APSIPA2026}
\usepackage{amsmath,amssymb,amsfonts,amsthm,bm,microtype}
\usepackage{multirow}
\usepackage{threeparttable}
\usepackage[backend=biber,style=ieee]{biblatex}
\usepackage{algorithmic}
\usepackage{graphicx}
\usepackage{textcomp}
\usepackage{url}
\usepackage{booktabs}
\usepackage{balance}

\DeclareMathOperator{\blkdiag}{blkdiag}
\theoremstyle{definition}

\theoremstyle{plain}
\newtheorem{theorem}{Theorem}

\usepackage{geometry}
\usepackage{fancyhdr}

\fancypagestyle{firststyle}{
  \fancyhf{}
  \fancyhead[C]{2026 Asia Pacific Signal and Information Processing Association Annual Summit and Conference (APSIPA ASC)}
}

\def\BibTeX{{\rm B\kern-.05em{\sc i\kern-.025em b}\kern-.08em
    T\kern-.1667em\lower.7ex\hbox{E}\kern-.125emX}}

\begin{document}

\title{Fast Multichannel NMF with Block-Diagonal Spatial Covariance Matrices for Efficient Blind Source Separation Using Distributed Microphone Arrays}

\author{
\authorblockN{
Hirotaka Nishikori\authorrefmark{1},
Nobutaka Ito\authorrefmark{2},
Kouei Yamaoka\authorrefmark{1},
Norihiro Takamune\authorrefmark{1}, and
Hiroshi Saruwatari\authorrefmark{1}
}

\authorblockA{
\authorrefmark{1}
The University of Tokyo, Tokyo, Japan
}

\authorblockA{
\authorrefmark{2}
The National Institute of Advanced Industrial Science and Technology (AIST), Tokyo, Japan
}
}


\maketitle
\begingroup\renewcommand\thefootnote{}\footnotetext{This work was supported by JSPS KAKENHI Grant Numbers 24K23854, 25K21220 and 21H05054.}\endgroup
\thispagestyle{firststyle}
\pagestyle{empty}

\begin{abstract}
Distributed microphone arrays composed of multiple subarrays enable blind source separation over a wide spatial area.
Directly applying fast multichannel nonnegative matrix factorization (FastMNMF) to all subarrays can exploit observations from all subarrays, but it requires repeated inversions
of large matrices spanning all microphones, causing the computational cost to increase rapidly as the number of microphones grows. 
In contrast, applying FastMNMF to one subarray reduces the matrix size but cannot exploit observations from other subarrays. 
 We propose \textit{distributed FastMNMF}, which imposes a block-diagonal structure on the source spatial 
covariance matrices, so that matrix inversions are performed within subarrays.
The NMF-based source spectrogram model is shared across subarrays, allowing the method to aggregate source activity information while discarding inter-subarray covariance.
 In synchronized, noiseless simulations with fixed room and array/source geometry,
 the method required less computation time than conventional FastMNMF using all subarrays, achieved a higher average source-to-distortion ratio than conventional FastMNMF using one subarray,
 and was applicable in the tested five-source condition, where each four-microphone subarray was locally underdetermined.
\end{abstract}


\section{Introduction}
\vspace{-5pt}
Blind source separation (BSS) \cite{Makino2018} aims to separate source signals from observed mixtures without detailed prior information.
 Independent vector analysis (IVA) \cite{Hiroe2006, Kim2007} and independent low-rank matrix analysis (ILRMA) \cite{ILRMA} are efficient in determined or overdetermined conditions,
 where microphones are at least as many as sources.
 However, they are generally not applicable to underdetermined mixtures,
 where there are more sources than microphones, and are not well suited to modeling diffuse noise.
 Multichannel nonnegative matrix factorization (MNMF) \cite{Ozerov2010, Sawada2013} can handle underdetermined conditions and model diffuse noise
 by using full-rank spatial covariance matrices (SCMs), but requires repeated matrix inversions, which are extremely costly.
 FastMNMF \cite{Ito2019,Sekiguchi2019,Ikeshita2018} reduces this cost while maintaining separation performance comparable to that of MNMF by assuming jointly diagonalizable SCMs.


While conventional array signal processing typically assumes a single compact array,
 distributed arrays enable BSS over a wide area using spatially separated microphones or subarrays,
 but practical use requires synchronization/calibration \cite{Masuyama2024} and often microphone clustering \cite{Kindt2023}.
 Moreover, designing BSS algorithms for distributed arrays is itself challenging because such algorithms must
 exploit information across many subarrays while remaining robust to less reliable inter-subarray phase relations and avoiding prohibitive computational
 cost as the number of microphones grows \cite{Sumura2024,Yamaoka2025}.
 Here we focus on reducing the computational cost of distributed-array BSS under the assumption that synchronization/calibration and microphone clustering
 have been addressed.
 In this paper, ``distributed'' refers to the array geometry composed of spatially separated subarrays.
 The proposed algorithm is not a decentralized one with explicit communication constraints but a centralized one using observations from all subarrays.
 Evaluating the proposed BSS method under synchronization/calibration errors and extending it to joint microphone clustering are left for future work.

 To reduce the computational cost of distributed-array BSS, we propose \textit{distributed FastMNMF} by imposing a block-diagonal structure on the source SCMs,
 where each block corresponds to one subarray. A naive extension of FastMNMF is to process all subarrays jointly as a single large array,
 which can exploit information across subarrays but incurs rapidly increasing computational cost as the number of microphones grows.
 Another option is to process only one subarray, which is computationally efficient but yields limited separation performance because
 it cannot exploit information from other subarrays. Unlike FastMNMF using all subarrays, the proposed method performs joint diagonalization within each subarray,
 reducing the sizes of the matrices to be diagonalized/inverted.
 Unlike FastMNMF using one subarray, it shares the source spectrograms across subarrays,
 thereby aggregating source spectrogram information while discarding inter-subarray covariance.
 The goal is not to outperform FastMNMF using all subarrays, but to provide a computationally efficient intermediate model between FastMNMF using
 all subarrays and that using one subarray.

Compared with distributed MNMF-based methods \cite{Sumura2024}, the proposed method uses joint diagonalizability and block-diagonality for efficiency.
 Unlike decentralized IVA \cite{Yamaoka2025}, it can handle locally underdetermined subarrays, and unlike transfer-function-gain NMF \cite{Togami2010,TFGNMF_chiba2014}, it can exploit within-subarray phase information.

\vspace{-5pt}
\section{Preliminaries}
\vspace{-2pt}
\subsection{FastMNMF \cite{Ito2019,Sekiguchi2019,Ikeshita2018}}
\vspace{-2pt}
Suppose $N$ source signals are mixed and observed by $M$ microphones.
Here, $m \in \{1,\dots,M\}$ and $n \in \{1,\dots,N\}$ denote the microphone and source indices, respectively.
Let $\bm{x}_{ij} = (x_{ij1}, \dots, x_{ijM})^{\mathsf{T}}\in \mathbb{C}^M$ denote the short-time Fourier transform (STFT) coefficients of the observed signals,
 where $i \in \{1,\dots,I\}$ and $j \in \{1,\dots,J\}$ are the frequency-bin and time-frame indices, respectively, and $\cdot^{\mathsf{T}}$ represents
 the transpose.
In addition, $\cdot^{\mathsf{H}}$, $\det$, and $\ln$ denote the Hermitian transpose, determinant, and natural logarithm, respectively.

$\bm{x}_{ij}$ is modeled as the sum of source images $\bm{c}_{ijn}$,
 each of which represents the contribution of a source to all microphones.
Assuming that each $\bm{c}_{ijn}$ follows a multivariate complex Gaussian distribution with zero mean and covariance matrix
$h_{ijn}\bm{R}_{in}$ and that $\{\bm{c}_{ijn}\}_n$ are mutually statistically independent, the closure property of the multivariate complex Gaussian
distribution yields the generative model of $\bm{x}_{ij}$:
\begin{align}
p(\bm{x}_{ij}) 
&= \mathcal{N}_{\mathbb{C}}\bigg(\bm{x}_{ij}; \bm{0}, \sum_n h_{ijn}\bm{R}_{in}\bigg).
\label{eq:generativemodel}
\end{align}
Here, $h_{ijn}\in \mathbb{R}_{\geq 0}$ and $\bm{R}_{in}\in \mathbb{S}^{M}_{+}$ denote the source spectrogram and the SCM of each source, respectively.
$\mathbb{R}_{\geq 0}$ is the set of nonnegative real numbers and $\mathbb{S}^{M}_{+}$ is the set of $M\times M$
complex Hermitian positive semidefinite matrices.
Note that the observation covariance matrix $\sum_n h_{ijn}\bm{R}_{in}$ is assumed to be positive definite so that \eqref{eq:generativemodel} is well-defined.
Furthermore, $h_{ijn}$ is modeled using the NMF model 
\begin{align}
  h_{ijn} = \sum_{k} t_{ikn}v_{kjn},
  \label{eq:nmf}
\end{align}
where $t_{ikn}, v_{kjn}\in \mathbb{R}_{\geq 0}$ denote the spectral basis and its temporal activation,
 respectively, and $k\in \{1,\dots, K \}$ is the basis index.

MNMF \cite{Sawada2013} does not impose a particular constraint on $\bm{R}_{in}$, whereas
 FastMNMF assumes joint diagonalizability of $\bm{R}_{in}$ across all sources to reduce computational cost: 
\begin{equation}
 \bm{W}_i^{\mathsf{H}} \bm{R}_{in} \bm{W}_i = \bm{\Lambda}_{in}, \qquad \forall n=1,\dots,N,
\label{eq:joint}
\end{equation}
where $\bm{W}_i=(\bm{w}_{i1},\dots,\bm{w}_{iM})\in \mathbb{C}^{M\times M}$ is a nonsingular transformation matrix and $\bm{\Lambda}_{in}\in \mathbb{S}^{M}_{+}$ is diagonal.
Then, $\bm{y}_{ij}=\bm{W}_{i}^{\mathsf{H}} \bm{x}_{ij}$ follows a multivariate complex Gaussian distribution with zero mean and covariance matrix $\sum_nh_{ijn}\bm{\Lambda}_{in}$.
Since $\bm{W}_i$ decorrelates the observed signals,
 we call $\bm{y}_{ij}=(y_{ij1},\dots, y_{ijM})^{\mathsf{T}}\in\mathbb{C}^M$ the decorrelated observed signals.

The negative log-likelihood up to an additive constant is
\begin{equation}
\begin{aligned}
&\sum_{i,j,m}
\biggl(
\frac{\lvert y_{ijm} \rvert^{2}}
{\sum_{k,n} t_{ikn} v_{kjn} [\bm{\Lambda}_{in}]_{mm}}
+
\ln\sum_{k,n} t_{ikn} v_{kjn} [\bm{\Lambda}_{in}]_{mm}
\biggr)\\
&-
\sum_{i} J \ln \lvert \det \bm{W}_i \rvert^{2},
\end{aligned}
\label{eq:cost_fastmnmf}
\end{equation}
where $[\bm{\Lambda}_{in}]_{mm}$ denotes the $(m,m)$th element of $\bm{\Lambda}_{in}$.

Sekiguchi {\it et al.} \cite{Sekiguchi2019} proposed update rules based on iterative projection (IP) \cite{IP1} and majorization-minimization (MM) algorithm,
 which guarantee a monotonic non-increase of \eqref{eq:cost_fastmnmf}.
 The update rules for $\bm{W}_i$ based on IP are
\begin{equation}
\!\!\!\!\!\!\!\!\!\!\!\!\!\!\!\!\!
\bm{Q}_{im}
\gets
\frac{1}{J}
\sum_j
\eta_{ijm}^{-1}
\bm{x}_{ij} \bm{x}_{ij}^{\mathsf{H}},
\label{eq:Q_update_f}
\end{equation}
\begin{equation}
\!\!\!\!\!\!\!\!\!\!\!\!\!\!\!\!\!
\bm{w}_{im}
\gets
\left(
\bm{W}_i^{\mathsf{H}} \bm{Q}_{im}
\right)^{-1}
\bm{e}_m,
\label{eq:w_update_f}
\end{equation}
\begin{equation}
\bm{w}_{im}
\leftarrow
\bm{w}_{im}
\left(
\bm{w}_{im}^{\mathsf{H}}
\bm{Q}_{im}
\bm{w}_{im}
\right)^{-\frac{1}{2}},
\label{eq:w_normalize_f}
\end{equation}
where $\bm{e}_m\in \mathbb{R}^{M}$ denotes the $m$th column vector of the $M$-dimensional identity matrix and $\eta_{ijm} = \sum_{k,n} t_{ikn}v_{kjn} [\bm{\Lambda}_{in}]_{mm}$.
The update rules for $t_{ikn}$, $v_{kjn}$, and $[\bm{\Lambda}_{in}]_{mm}$ based on the MM algorithm are 
\begin{equation}
\quad \quad \, t_{ikn} \leftarrow t_{ikn} \sqrt{\frac{\sum_{j,m} v_{kjn} [\bm{\Lambda}_{in}]_{mm} \lvert y_{ijm}\rvert^{2} \eta_{ijm}^{-2}}{\sum_{j,m} v_{kjn} [\bm{\Lambda}_{in}]_{mm} \eta_{ijm}^{-1}}}, \\
\label{eq:t_f}
\end{equation}
\begin{equation}
\quad \:\:\:\; v_{kjn} \leftarrow v_{kjn} \sqrt{\frac{\sum_{i,m} t_{ikn} [\bm{\Lambda}_{in}]_{mm} \lvert y_{ijm}\rvert^{2} \eta_{ijm}^{-2}}{\sum_{i,m} t_{ikn} [\bm{\Lambda}_{in}]_{mm} \eta_{ijm}^{-1}}}, \\
\label{eq:v_f}
\end{equation}
\begin{equation}
[\bm{\Lambda}_{in}]_{mm} \leftarrow [\bm{\Lambda}_{in}]_{mm} \sqrt{\frac{\sum_{j,k} t_{ikn}v_{kjn} \lvert y_{ijm}\rvert^{2} \eta_{ijm}^{-2}}{\sum_{j,k} t_{ikn}v_{kjn} \eta_{ijm}^{-1}}}.
\label{eq:Lambda_f}
\end{equation}
Here, $\eta_{ijm}$ is updated after every update in \eqref{eq:t_f}--\eqref{eq:Lambda_f}.
The parameters are estimated by alternately applying \eqref{eq:Q_update_f}--\eqref{eq:Lambda_f} for a fixed number of iterations.
After that, the source images are estimated using the multichannel Wiener filter \cite{Ito2019}.

FastMNMF reduces computational cost by assuming the joint diagonalizability.
MNMF based on the Itakura--Saito divergence \cite{Sawada2013} requires inverting $M\times M$ matrices for each time-frequency point and each iteration, and solving $IN$ algebraic Riccati equations per iteration. Both are expensive operations with complexity $\mathcal{O}(M^3)$.
In contrast, FastMNMF requires inversions of $M\times M$ matrices only for each frequency bin and each microphone per iteration, and does not require solving Riccati equations.
 However, the remaining $IM$ inversions can still be costly for large $M$.

 \vspace{-2pt}
\subsection{Naive application of FastMNMF to distributed arrays}
\vspace{-2pt}
A simple way to apply FastMNMF to distributed arrays is to apply it to all subarrays jointly.
In this case, $h_{ijn}$ and $\bm{R}_{in}$ are estimated jointly from all subarrays.
$\bm{R}_{in}$ has size $M\times M$,
where $M$ is the total number of microphones across all subarrays.
Although this method can exploit information across subarrays, its computational cost increases rapidly as the number of subarrays grows.
\eqref{eq:w_update_f} requires $IM$ inversions of $M\times M$ matrices.
Since each matrix inversion costs $\mathcal{O}(M^3)$, this part costs $\mathcal{O}(IM^4)$ per iteration.
This term grows rapidly as $M$ increases. The computational complexity per iteration and per frequency is shown in Table~\ref{tab:complexity_breakdown}.

In contrast, FastMNMF can be applied to one subarray (let it be the $l$th one).
In this case, the $l$th subarray estimates the source spectrograms $h_{ijn}$ and the SCMs $\bm{R}_{in}^{(l)}$ of size $M^{(l)}\times M^{(l)}$ independently of the other subarrays.
Here, $\bm{R}_{in}^{(l)}$ is the SCM of the $n$th source for the $l$th subarray, and $M^{(l)}$ denotes the number of microphones in the $l$th subarray.
Although this method reduces the computational complexity per iteration and per frequency to $\mathcal{O}(M^{(l)4}+JM^{(l)3}+JN(K+M^{(l)}))$, its separation performance can be limited because it cannot exploit information from other subarrays.
These methods are used as baselines in Section~IV.

\vspace{-5pt}
\section{Proposed Method}
\vspace{-2pt}
\subsection{Motivation and approach}
\vspace{-2pt}
Distributed FastMNMF imposes a block-diagonality constraint on the SCMs in addition to the joint diagonalizability constraint, with each block corresponding to a subarray.
The block-diagonal SCM should not be interpreted as the true physical SCM.
 Since the same source signal reaches multiple subarrays, the true source images can generally have nonzero
 cross-subarray covariance, provided that the subarrays are synchronized.
 The block-diagonal SCM model is therefore a computational approximation introduced for tractability.
 It may also be useful when cross-subarray spatial information is unreliable,
 but evaluating such nonideal conditions is beyond the scope of this paper.
As such, the proposed model discards inter-subarray covariance and phase relations,
 while sharing only source spectrograms modeled with NMF.
  Sharing source spectrograms is most appropriate when inter-subarray propagation delays are sufficiently short relative to the STFT window length.
 We assume that inter-subarray calibration and sampling synchronization have already been performed.

\vspace{-2pt}
\subsection{Distributed FastMNMF}
\vspace{-2pt}
Distributed FastMNMF is obtained from conventional FastMNMF by imposing on the source SCMs an additional block-diagonal structure with respect to the subarray partition.
We consider a distributed array consisting of $L$ subarrays.
Let $l \in \{1,\dots,L\}$ denote the subarray index and $M^{(l)}$ be the number of microphones in the $l$th subarray so that $M = \sum_{l} M^{(l)}$.
The microphones are indexed so
that those belonging to the same subarray are contiguous. Specifically, let 
$\bm{x}_{ij}=(\bm{x}_{ij}^{(1)\textsf{T}},\dots,\bm{x}_{ij}^{(L)\textsf{T}})^\textsf{T}$, where $\bm{x}_{ij}^{(l)}=(x_{ij1}^{(l)},\dots,x_{ijM^{(l)}}^{(l)})^\textsf{T}$ is the observation vector of the $l$th subarray.

The SCMs $\bm{R}_{in}$ are assumed to be block-diagonal with each block corresponding to a subarray:
\begin{align}
\bm{R}_{in} &= \blkdiag\left(\bm{R}_{in}^{(1)},\dots,\bm{R}_{in}^{(L)}\right),
\end{align}
where $\bm{R}_{in}^{(l)} \in \mathbb{S}^{M^{(l)}}_{+}$ denotes the SCM of the $n$th source at the $l$th subarray
 and the operator $\blkdiag$ constructs a block-diagonal matrix by arranging its matrix arguments.
In addition, \eqref{eq:joint} is assumed.
Under a positive-definiteness condition in the Appendix, this is equivalent to the joint diagonalizability of $\{\bm{R}^{(l)}_{in}\}_n$ for all $l$:
\begin{equation}
\bm{W}_i^{(l)\mathsf{H}} \bm{R}_{in}^{(l)} \bm{W}_i^{(l)} = \bm{\Lambda}_{in}^{(l)}, \qquad \forall n=1,\dots,N.
\label{eq:joint_diagonalizability}
\end{equation}
The matrices $\bm{W}_i^{(l)}=(\bm{w}_{i1}^{(l)},\dots,\bm{w}_{iM^{(l)}}^{(l)})\in \mathbb{C}^{M^{(l)} \times M^{(l)}}$ and  the diagonal matrices $\bm{\Lambda}_{in}^{(l)}\in \mathbb{S}^{M^{(l)}}_{+}$ are defined for each subarray.
As in conventional FastMNMF, the source spectrograms are modeled by \eqref{eq:nmf} and shared across all subarrays. Alternatively,
 the distributed FastMNMF model can also be interpreted as a variant of the FastMNMF model with the additional constraint that $\bm{W}_i$
 is block-diagonal $\bm{W}_i = \blkdiag(\bm{W}_i^{(1)},\dots,\bm{W}_i^{(L)})$. It minimizes the same cost function \eqref{eq:cost_fastmnmf} as conventional FastMNMF under this constraint.

The cost function based on the negative log-likelihood of the observed signals under the above model is given by
\begin{equation}
\begin{aligned}
\! &\sum_l\! \Biggl[
    \sum_{i,j,\mu}\!
    \Biggl(\!
        \frac{\lvert y_{ij\mu}^{(l)} \rvert^2}
        {\sum_{k,n}\! t_{ikn} v_{kjn} [\bm{\Lambda}_{in}^{(l)}]_{\mu\mu}}
        \!+\! \ln\! \sum_{k,n} t_{ikn} v_{kjn} [\bm{\Lambda}_{in}^{(l)}]_{\mu\mu}
    \!\Biggr) \\
&\:\:\:\quad\quad
    - \sum_i J \ln \lvert \det \bm{W}_i^{(l)} \rvert^2
\Biggr],
\end{aligned}
\label{eq:cost_dfastmnmf}
\end{equation}
where $\mu \in \{1,\dots,M^{(l)}\}$ is the microphone index within each subarray and
$\bm{y}_{ij}^{(l)}=(y_{ij1}^{(l)},\dots,y_{ijM^{(l)}}^{(l)})^{\mathsf{T}} = \bm{W}_{i}^{(l)\mathsf{H}} \bm{x}_{ij}^{(l)}$
are the decorrelated observed signals in the $l$th subarray.

The parameters $\{t_{ikn}, v_{kjn}, \bm{W}_i^{(l)}, \bm{\Lambda}_{in}^{(l)}\}$ are estimated by minimizing \eqref{eq:cost_dfastmnmf}.
Since, for each $l$, the term inside the square brackets in \eqref{eq:cost_dfastmnmf} has the same form as
 \eqref{eq:cost_fastmnmf} with respect to $\bm{W}_i^{(l)}$, 
 IP can be applied independently to each subarray:
\begin{equation}
\!\!\!\!\!\!\!\!\!\!
\bm{Q}_{i\mu}^{(l)}
\gets
\frac{1}{J}
\sum_j
\eta_{ij\mu}^{(l)-1}
\bm{x}_{ij}^{(l)} \bm{x}_{ij}^{(l)\mathsf{H}},
\label{eq:Q_update}
\end{equation}
\begin{equation}
\!\!\!\!\!\!\!\!\!\!\!\!\!\!\!
\bm{w}_{i\mu}^{(l)}
\gets
\left(
\bm{W}_i^{(l)\mathsf{H}} \bm{Q}_{i\mu}^{(l)}
\right)^{-1}
\bm{e}_\mu,
\label{eq:w_update}
\end{equation}
\begin{equation}
\bm{w}_{i\mu}^{(l)}
\leftarrow
\bm{w}_{i\mu}^{(l)}
\left(
\bm{w}_{i\mu}^{(l)\mathsf{H}}
\bm{Q}_{i\mu}^{(l)}
\bm{w}_{i\mu}^{(l)}
\right)^{-\frac{1}{2}},
\label{eq:w_normalize}
\end{equation}
where $\bm{e}_\mu$ denotes the $\mu$th column of the
 $M^{(l)}\times M^{(l)}$ identity matrix and $\eta_{ij\mu}^{(l)} = \sum_{k,n} t_{ikn}v_{kjn} [\bm{\Lambda}_{in}^{(l)}]_{\mu\mu}$.

Next, we derive the update rules for $t_{ikn}$, $v_{kjn}$, and $\bm{\Lambda}_{in}^{(l)}$.
As already mentioned, distributed FastMNMF minimizes the same cost \eqref{eq:cost_fastmnmf} as conventional FastMNMF under the additional block-diagonality constraint on $\bm{W}_i$. Therefore, with $\bm{W}_i$ fixed,
$t_{ikn}$, $v_{kjn}$, and $\bm{\Lambda}_{in}=\blkdiag(\bm{\Lambda}_{in}^{(1)},\dots,\bm{\Lambda}_{in}^{(L)})$ can be updated by the conventional
update rules \eqref{eq:t_f}--\eqref{eq:Lambda_f}, which guarantee monotonic non-increase of \eqref{eq:cost_fastmnmf}.
Here, we set $y_{ijm} = y_{ij\mu}^{(l)}$ and $\eta_{ijm} = \eta_{ij\mu}^{(l)}$, where $m=\mu+\sum_{\lambda=1}^{l-1}M^{(\lambda)}$.

Parameters are estimated by alternately applying \eqref{eq:Q_update}--\eqref{eq:w_normalize} and \eqref{eq:t_f}--\eqref{eq:Lambda_f} for a fixed number of iterations.
After that, the source images at each subarray $\bm{c}_{ijn}^{(l)}$ are estimated by applying the multichannel Wiener filter to each subarray.

Table~\ref{tab:complexity_breakdown} summarizes the computational cost per parameter-estimation iteration and per frequency bin.
Multiplying the per-frequency-bin costs in Table~\ref{tab:complexity_breakdown} by $I$, in FastMNMF (all subarrays),
\eqref{eq:Q_update_f} requires $IJM$ scalar-matrix multiplications of $M\times M$ matrices, resulting in $\mathcal{O}(IJM^3)$, whereas  \eqref{eq:w_update_f} requires $IM$ inversions of $M\times M$ matrices, resulting in $\mathcal{O}(IM^4)$.
 Both terms increase rapidly as $M$ grows.
The proposed method reduces them to $\mathcal{O}(IJ\sum_l M^{(l)3})$ and $\mathcal{O}\!\left(I\sum_l M^{(l)4}\right)$ by imposing the block-diagonality constraint.
Computing the numerators and denominators in \eqref{eq:t_f}--\eqref{eq:Lambda_f}, together with updating $\eta_{ijm}$, is common to both methods.
Since no variable depends simultaneously on both $k$ and $m$, the summations in these updates can be factorized as $\eta_{ijm} = \sum_n\left(\sum_k t_{ikn}v_{kjn}\right)[\bm{\Lambda}_{in}]_{mm}$, 
 which costs $\mathcal{O}(IJN(K+M))$.
Computing $y_{ijm}$ costs $\mathcal{O}(IJM^2)$ for FastMNMF (all subarrays) and $\mathcal{O}(IJ\sum_l M^{(l)2})$ for the proposed method.
\begin{table}[t]
\centering
\caption{Computational complexity of parameter estimation per iteration and per frequency bin.}
\label{tab:complexity_breakdown}
\vspace{-9pt}
\resizebox{\columnwidth}{!}{%
\begin{tabular}{c|cc}
\toprule
 & FastMNMF (all subarrays) & Distributed FastMNMF \\
\midrule
$\bm{W}_{i}$ & $\mathcal{O}\left(M^{4}\!+\!JM^{3}\!\right)$ & $\mathcal{O}\!\left(\sum_l\! M^{(l)4}\!+\!J\sum_l\! M^{(l)3}\!\right)$ \\[3pt]
$t_{ikn}, v_{kjn}, \bm{\Lambda}_{in}\!$ & $\mathcal{O}\left(JM^{2}\!+\!JN\left(K\!+\!M\right)\right)$ & $\mathcal{O}\left(J\sum_l\!M^{(l)2}\!+\!JN\left(K\!+\!M\right)\right)$ \\
\midrule
Total & $\mathcal{O}\left(M^{4}\!+\!JM^{3}\!+\!JN\left(K\!+\!M\right)\right)$ & $\mathcal{O}\!\left(\sum_l\! M^{(l)4}\!+\!J\sum_l\! M^{(l)3}\!+\!JN\left(K\!+\!M\right)\right)$ \\
\bottomrule
\end{tabular}
}
\vspace{-15pt}
\end{table}

Distributed FastMNMF reduces to conventional FastMNMF when $L=1$ and is similar to the transfer-function-gain NMF \cite{TFGNMF_chiba2014}
 when $L=M$, where $\bm{R}_{in}$ is diagonal and models only the amplitudes of the observed signals.

\vspace{-5pt}
\section{Experiments}
\vspace{-2pt}
\subsection{Experimental conditions}
\vspace{-2pt}

As a preliminary evaluation of the proposed method, we conducted experiments simulating distributed-array BSS.
All algorithms were implemented in Python 3.12.7 with NumPy 1.26.0, SciPy 1.14.1, and Scikit-learn 1.8.0.

We used Pyroomacoustics 0.8.4 to generate room impulse responses
for a room with dimensions of $6\,\mathrm{m} \times 4\,\mathrm{m} \times 2.5\,\mathrm{m}$,
in which three subarrays and either three or five point sources were placed at a height of $1.5\,\mathrm{m}$ as illustrated in Fig.~\ref{fig:room}.
The centroids of the three subarrays were fixed at $(2, 2)$\,m, $(3, 2)$\,m, and $(4, 2)$\,m.
Each subarray consisted of four microphones at the vertices of a regular tetrahedron with an edge length of $4.2$\,cm.
The base of each tetrahedron was parallel to the floor.
One base edge of the left subarray was parallel to the $x$-axis, and the middle and right subarrays were rotated clockwise by $45^\circ$ and $90^\circ$, respectively.
For the three-source case, the source coordinates were $(1, 1)$\,m, $(3, 3.5)$\,m, and $(5, 1)$\,m.
For the five-source case, two sources at $(1.5, 3)$\,m and $(4.5, 3)$\,m were added.
Note that the three- and five-source cases correspond to determined and underdetermined
 conditions for each subarray, respectively.
The reverberation time was set to $\mathrm{RT}_{60}=300$\,ms, and the wall energy absorption and the maximum order of the image source method were computed using \texttt{pyroomacoustics.inverse\_sabine}.

As dry sources, we used speech signals from the JNAS corpus \cite{jnas}, adjusted to $10$\,s by truncating or repeating them.
For the three-source and five-source conditions, we generated $120$ mixtures by randomly selecting speech files so that no mixture contained the same speaker or utterance twice.
 Each possible male--female composition was equally represented across mixtures.
 The dry sources were convolved with the room impulse responses, and mixed so that all source images had the same power at a reference microphone in the left subarray in Fig.~\ref{fig:room}.
 The sampling frequency was $16$\,kHz.
A Hann window of $256$\,ms with a $64$\,ms shift, selected based on preliminary experiments, was used for the STFT.
No additive noise or sampling asynchrony was considered.

We compared FastMNMF (all subarrays) \cite{Sekiguchi2019} using all $12$ microphones,
 FastMNMF (one subarray) using the left subarray in Fig.~\ref{fig:room}, and distributed FastMNMF using all subarrays.
 All methods used $K=16$ and $200$ iterations.
All parameters in denominators and logarithms are optimized over the strictly positive domain.
The NMF variables and diagonal SCM entries are initialized with positive values, and the transformation matrices are initialized to be nonsingular.
The updates of $\eta_{ijm}$, the denominators in \eqref{eq:t_f}--\eqref{eq:Lambda_f}, and the normalization in \eqref{eq:w_normalize_f}, \eqref{eq:w_normalize} were floored at $10^{-6}$ to avoid division by zero.
 If $\bm{Q}_{im}$ is singular, the pseudo-inverse is used for \eqref{eq:w_update_f}.

\begin{figure}[t]
\centering
\includegraphics[width=0.6\columnwidth]{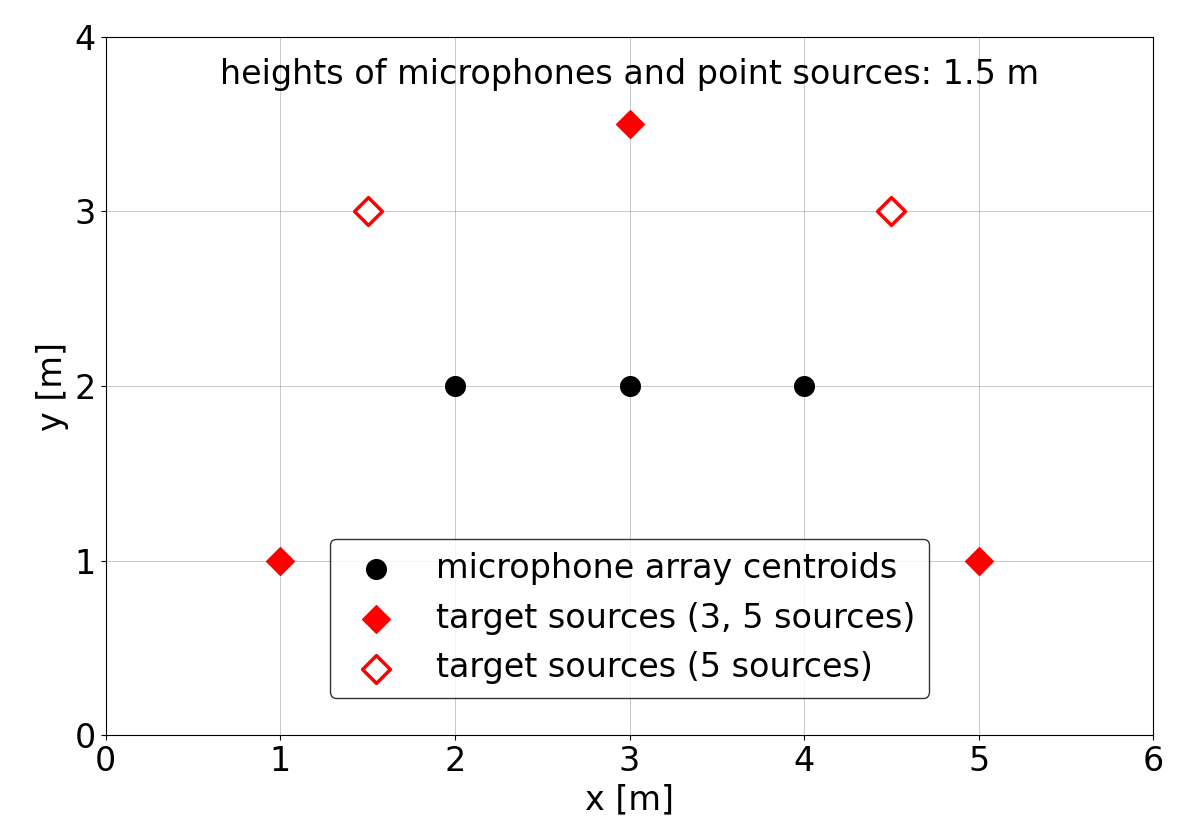}
\vspace{-9pt}
\caption{Room configuration for room impulse response generation.}
\label{fig:room}
\vspace{-15pt}
\end{figure}

The parameters were initialized as follows.
1) Time-frequency masks were estimated by frequency bin-wise clustering, followed by frequency-permutation alignment that combines local optimization
 with global optimization using one centroid per source \cite{Sawada2011}.
FastMNMF (all subarrays) used observations from all subarrays, while FastMNMF (one subarray) used only the left subarray $l=1$ in Fig.~\ref{fig:room} for mask estimation.
Distributed FastMNMF estimated masks independently in each subarray and then aligned the inter-subarray permutation by maximizing the sum of soft mask correlation coefficients across subarrays.
2) Initial source images were obtained by soft masking. Distributed FastMNMF and FastMNMF (all subarrays) used
$\bm{c}_{ijn}^{(\mathrm{init})}=(\bm{c}_{ijn}^{(\mathrm{init})(1)\mathsf{T}}, \bm{c}_{ijn}^{(\mathrm{init})(2)\mathsf{T}}, \bm{c}_{ijn}^{(\mathrm{init})(3)\mathsf{T}})^{\mathsf{T}}$, while FastMNMF (one subarray) used only $\bm{c}_{ijn}^{(\mathrm{init})(1)}$.
3) Initial source spectrograms were computed by
$h_{ijn}^{(\mathrm{init})} = \bm{c}_{ijn}^{(\mathrm{init})\mathsf{H}}\bm{R}_{in}^{(\mathrm{init})-1}\bm{c}_{ijn}^{(\mathrm{init})}\!/\!M$ 
with $\bm{R}_{in}^{(\mathrm{init})}\!=\!\sum_j\bm{c}_{ijn}^{(\mathrm{init})}\bm{c}_{ijn}^{(\mathrm{init})\mathsf{H}}\!/\!J$
 in FastMNMF (all subarrays) and distributed FastMNMF, and by
$h_{ijn}^{(\mathrm{init})}= \bm{c}_{ijn}^{(\mathrm{init})(1)\mathsf{H}}\bm{R}_{in}^{(\mathrm{init})(1)-1}\bm{c}_{ijn}^{(\mathrm{init})(1)}\!/\!M^{(1)}$
 with $\bm{R}_{in}^{(\mathrm{init})(1)}\!=\!\sum_j\bm{c}_{ijn}^{(\mathrm{init})(1)}\bm{c}_{ijn}^{(\mathrm{init})(1)\mathsf{H}}\!/\!J$
 in FastMNMF (one subarray). This procedure is motivated by maximum likelihood estimation under the time-varying complex Gaussian distribution.
4) Initial NMF variables $t_{ikn}^{(\mathrm{init})}$ and $v_{kjn}^{(\mathrm{init})}$ were computed by applying Itakura--Saito NMF to $h_{ijn}^{(\mathrm{init})}$.
\texttt{sklearn.decomposition.NMF} with \texttt{init=`random'}, \texttt{solver=`mu'}, \texttt{beta\_loss=`itakura-saito'}, and \texttt{max\_iter=1000} was used; all other options were left at their default values.
5) Initial transformation matrices were computed. In FastMNMF (all subarrays), $\bm{W}_{i}^{(\mathrm{init})}$ was computed by jointly diagonalizing
 $\bm{R}_{i,N-1}^{(\mathrm{init})}$ and $\bm{R}_{iN}^{(\mathrm{init})}$ via a generalized eigenvalue problem.
 Likewise, $\bm{W}_{i}^{(\mathrm{init})(l)}$ was computed from the $l$th diagonal blocks $\bm{R}_{i,N-1}^{(\mathrm{init})(l)}$ and $\bm{R}_{iN}^{(\mathrm{init})(l)}$
 of $\bm{R}_{i,N-1}^{(\mathrm{init})}$ and $\bm{R}_{iN}^{(\mathrm{init})}$ for each $l$ in distributed FastMNMF,
 and $\bm{W}_{i}^{(\mathrm{init})(1)}$ from $\bm{R}_{i,N-1}^{(\mathrm{init})(1)}$ and $\bm{R}_{iN}^{(\mathrm{init})(1)}$ in FastMNMF (one subarray).
 6) Initial diagonalized SCMs were computed by
 $\bm{\Lambda}_{in}^{(\mathrm{init})}=\operatorname{ddiag}(\bm{W}_{i}^{(\mathrm{init})\mathsf{H}} \bm{R}_{in}^{(\mathrm{init})} \bm{W}_{i}^{(\mathrm{init})})$
 in FastMNMF (all subarrays),
 $\bm{\Lambda}_{in}^{(\mathrm{init})(l)}=\operatorname{ddiag}(\bm{W}_{i}^{(\mathrm{init})(l)\mathsf{H}} \bm{R}_{in}^{(\mathrm{init})(l)} \bm{W}_{i}^{(\mathrm{init})(l)})$
 in distributed FastMNMF, and
 $\bm{\Lambda}_{in}^{(\mathrm{init})(1)}=\operatorname{ddiag}(\bm{W}_{i}^{(\mathrm{init})(1)\mathsf{H}} \bm{R}_{in}^{(\mathrm{init})(1)} \bm{W}_{i}^{(\mathrm{init})(1)})$
 in FastMNMF (one subarray). Here, $\operatorname{ddiag}$ is the projection onto the diagonal matrices.

The source-to-distortion ratio (SDR) improvement was computed at the reference microphone.
We used \texttt{fast\_bss\_eval.sdr} with \texttt{filter\_length=512}, which finds the global permutation by maximizing the sum of SDRs.\footnote{The code is available at \url{https://github.com/fakufaku/fast_bss_eval}.}
The reported mean SDR improvement was computed by averaging over sources, $10$ NMF initializations, and $120$ mixtures.
 Since the evaluation uses the fixed reference microphone, the results should be interpreted as output quality at that microphone rather than
 as an average over all microphones.

\vspace{-2pt}
\subsection{Experimental results}
\vspace{-2pt}

\begin{figure}[t]
\centering
\includegraphics[width=\columnwidth]{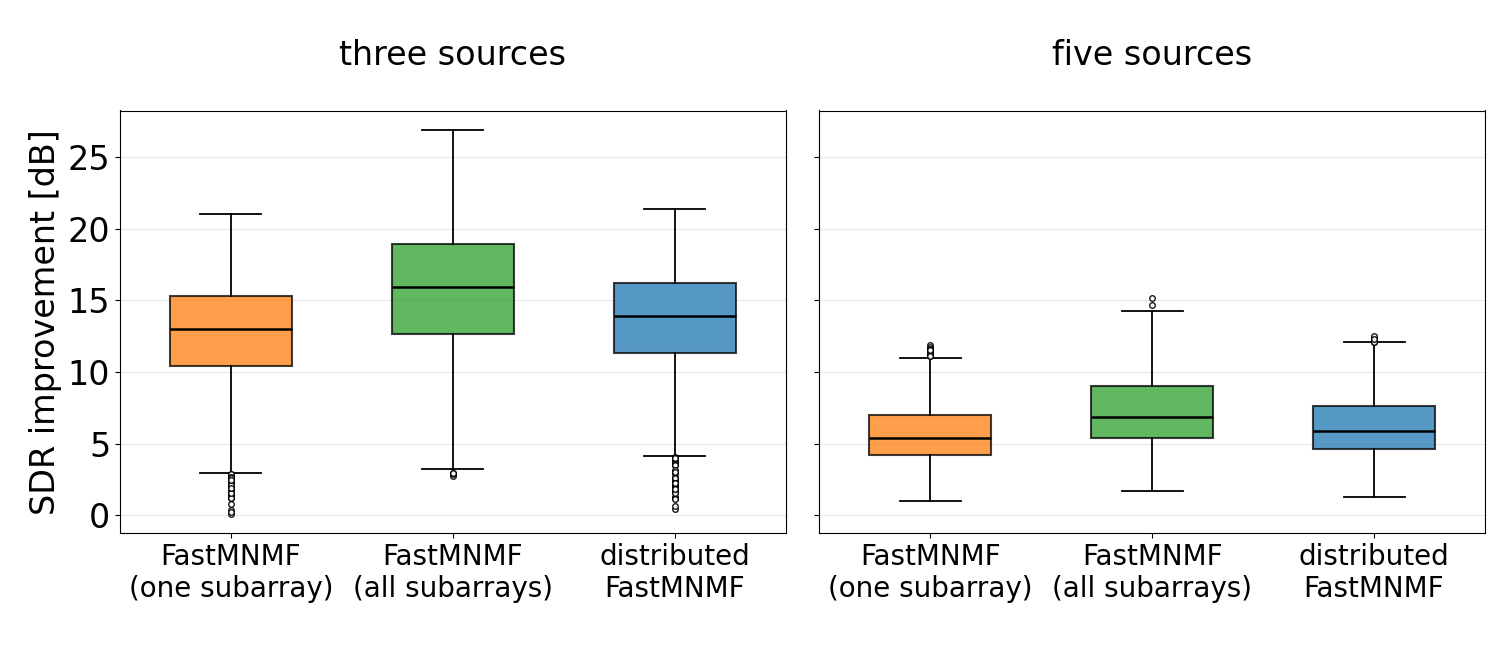}
\vspace{-25pt}
\caption{Box plots of the SDR improvement for each method and number of sources. The plots summarize the distributions over the evaluated mixtures and NMF initializations. Outliers are marked with circles.}
\label{fig:sdr}
\vspace{-15pt}
\end{figure}

Figure~\ref{fig:sdr} summarizes the SDR improvement distributions.
For the three-source case, distributed FastMNMF achieved a mean SDR improvement of $13.4$\,dB (median: $13.9$\,dB, standard error (SE): $0.114$\,dB),
 compared with $12.5$\,dB (median: $13.0$\,dB, SE: $0.110$\,dB) for FastMNMF (one subarray) and $15.7$\,dB (median: $15.9$\,dB, SE: $0.142$\,dB) for FastMNMF (all subarrays).
For the five-source case, the corresponding mean SDR improvements were $6.3$\,dB (median: $5.9$\,dB, SE: $0.064$\,dB), $5.8$\,dB (median: $5.4$\,dB, SE: $0.060$\,dB),
 and $7.3$\,dB (median: $6.9$\,dB, SE: $0.076$\,dB), respectively.
The SDR improvement was first averaged over sources for each mixture and NMF initialization.
 The reported mean and median summarize the evaluated mixtures and initializations.
 The standard error was computed over $1200$ mixture-initialization trials.

Distributed FastMNMF yielded higher average SDR improvement than FastMNMF (one subarray) for both source numbers.
 The average gains over FastMNMF (one subarray) were $0.8$\,dB for three sources and $0.5$\,dB for five sources.
 These results indicate that sharing the source spectrogram model across subarrays is beneficial in the tested setting.
 A complementary experiment under the same conditions as in Fig.~\ref{fig:sdr} showed that
 distributed FastMNMF with source spectrograms estimated independently per subarray yielded mean SDR improvements matching
 FastMNMF (one subarray) to machine precision across all 10 NMF initializations and 120 mixtures.
 Distributed FastMNMF yielded lower SDR improvement than FastMNMF (all subarrays), which is expected because it discards inter-subarray covariance
 and phase relations and applies the multichannel Wiener filter within each local subarray.
 Note that the five-source condition is underdetermined with respect to each four-microphone subarray, but not with respect to the full 12-microphone array.  




\begin{table}[t]
\centering
\caption{Computation time for each method in the three-source condition using a fixed $10$-s-long mixture and fixed NMF initialization seed (mean $\pm$ SE over 10 trials).}
\label{tab:time}
\vspace{-9pt}
\begin{tabular}{c|c}
\toprule
Method & Computation time [s] \\
\midrule
FastMNMF (one subarray)  & $109.3 \pm 0.3$ \\
FastMNMF (all subarrays) & $694.0 \pm 0.7$ \\
Distributed FastMNMF  & $235.3 \pm 2.4$ \\
\bottomrule
\end{tabular}
\end{table}

\begin{figure}[t]
\centering
\includegraphics[width=0.9\columnwidth]{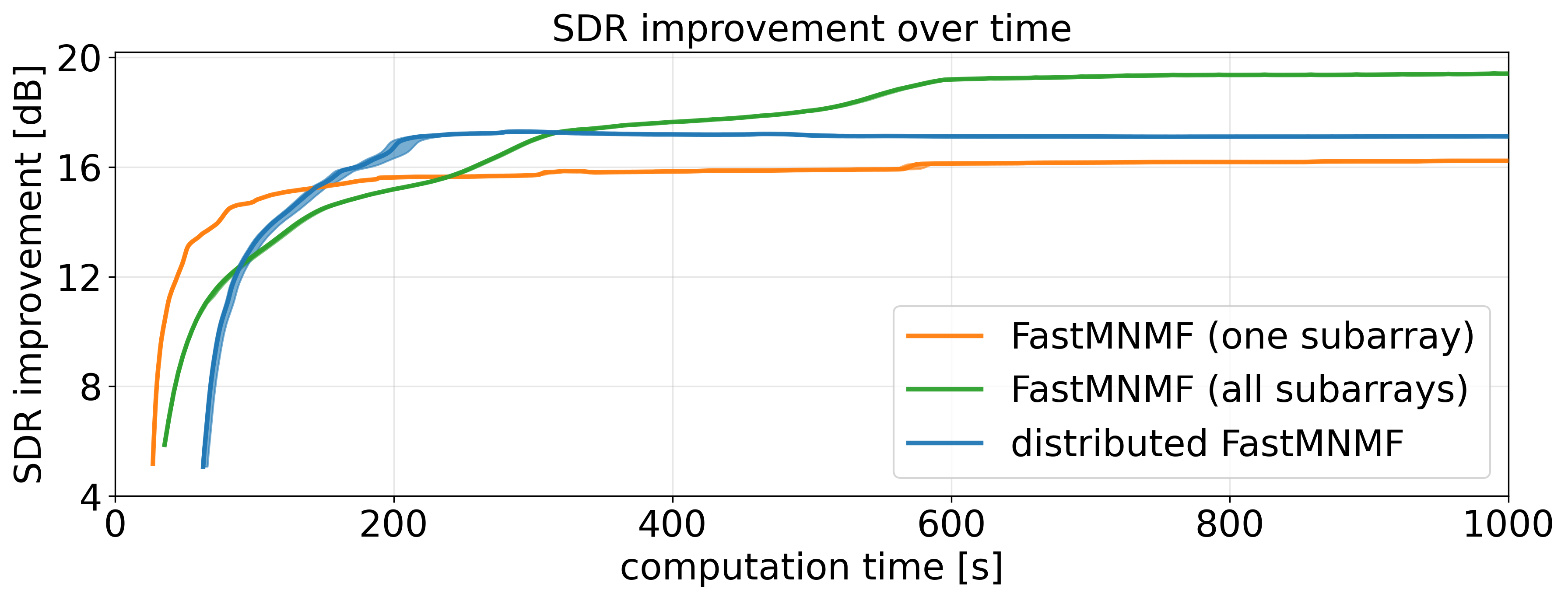}
\vspace{-15pt}
\caption{SDR improvement versus computation time for each method in the three-source condition using a fixed $10$-s-long mixture and a fixed NMF initialization seed, 
averaged over 10 trials. Shaded bands denote min-max computation time.}
\label{fig:sdr_time}
\vspace{-12pt}
\end{figure}

\vspace{-2pt}
\subsection{Evaluation of computation time}
\vspace{-2pt}
Computation time was measured over $10$ trials for the three-source condition, using a fixed $10$-s mixture, fixed NMF seed, and $200$ iterations for all methods.
 All methods were run on an AMD Ryzen 5 5600X ($3.7$\,GHz) processor using a single thread.
 These timings quantify the runtime in the present implementation and experimental setting.
 They are not intended as a scaling study over the number of subarrays, microphones, sources, or room conditions. 

Table~\ref{tab:time} shows the average computation time.
Distributed FastMNMF required $235.3$\,s on average, $33.9$\% of the average runtime of FastMNMF (all subarrays), corresponding to a $2.95\times$ speedup. It required $215$\% of that of FastMNMF (one subarray).

From Table~\ref{tab:complexity_breakdown}, if all subarrays have the same number of microphones, the proposed method reduces matrix inversion and scalar-matrix multiplication costs by factors of $L^3$ and $L^2$, respectively, relative to FastMNMF (all subarrays).
However, the $\mathcal{O}\left(JN\left(K+M\right)\right)$ part corresponding to the updates of the NMF variables and diagonalized SCMs remains unchanged.
 Therefore, with three subarrays, the speedup was smaller than the asymptotic 27-fold or 9-fold gain.

Figure~\ref{fig:sdr_time} shows SDR improvement over time, computed from the separated signals at each iteration.
 This evaluation step was excluded from the runtime.
At approximately 150--300\,s, distributed FastMNMF achieved higher SDR improvement than both FastMNMF (one subarray) and FastMNMF (all subarrays) before convergence, balancing the trade-off between SDR improvement and computation time.


\vspace{-5pt}
\section{Conclusion}
\vspace{-2pt}
We proposed distributed FastMNMF for distributed microphone arrays by imposing block-diagonality on the SCMs.
In experiments in a synchronized, noiseless simulated room, the proposed method was faster than FastMNMF (all subarrays)
 and improved SDR over FastMNMF (one subarray), including in a locally underdetermined condition.
 Future work will address diffuse noise, asynchronous recording, larger arrays, and wider reverberation and geometry conditions.

\section*{Appendix: Relationship between \eqref{eq:joint} and \eqref{eq:joint_diagonalizability}}
\label{appendix}

\begin{theorem}
Let $N,L,M^{(1)},\ldots,M^{(L)}\in\mathbb{Z}_{>0}$, $M:=\sum_l M^{(l)}$, and $\bm{R}_n=\blkdiag(\bm{R}_n^{(1)},\ldots,\bm{R}_n^{(L)})\in\mathbb{S}^M_+$,
 where $\bm{R}_n^{(l)}\in\mathbb{S}^{M^{(l)}}_+$ for $n=1,\ldots,N$ and $l=1,\ldots,L$. $\mathbb{Z}_{>0}$ is the set of positive integers.
 Let $\bm{S}:=\sum_n\bm{R}_n$ be positive definite.\footnote{This condition holds under the positive definiteness assumption on the observed covariance matrix.}
Then, $\bm{R}_1,\ldots,\bm{R}_N$ are jointly diagonalizable iff,
for every $l=1,\ldots,L$, $\bm{R}_1^{(l)},\ldots,\bm{R}_N^{(l)}$ are jointly diagonalizable.
\end{theorem}

{\setlength{\parindent}{-30pt}
\begin{proof}
 Define $\bm{S}^{(l)}:=\sum_n\bm{R}_n^{(l)}$, $\widetilde{\bm{R}}_n:=\bm{S}^{-1/2}\bm{R}_n\bm{S}^{-1/2}$,
 and $\widetilde{\bm{R}}_n^{(l)}:=(\bm{S}^{(l)})^{-1/2}\bm{R}_n^{(l)}(\bm{S}^{(l)})^{-1/2}$.
 Then $\bm{S}^{(l)}\succ\bm{O}$ and $\widetilde{\bm{R}}_n=\blkdiag(\widetilde{\bm{R}}_n^{(1)},\ldots,\widetilde{\bm{R}}_n^{(L)})$.
 By \cite{Wang2025}, $\bm{R}_1,\ldots,\bm{R}_N$ are jointly diagonalizable iff
 $[\widetilde{\bm{R}}_n,\widetilde{\bm{R}}_{n'}]:=\widetilde{\bm{R}}_n\widetilde{\bm{R}}_{n'}-\widetilde{\bm{R}}_{n'}\widetilde{\bm{R}}_n=\bm{O}, \quad\forall n,n'$. Since $[\widetilde{\bm{R}}_n,\widetilde{\bm{R}}_{n'}]=\blkdiag([\widetilde{\bm{R}}_n^{(1)},\widetilde{\bm{R}}_{n'}^{(1)}],\ldots,[\widetilde{\bm{R}}_n^{(L)},\widetilde{\bm{R}}_{n'}^{(L)}])$,
 the theorem has been proven.
\end{proof}
}

\vspace{-2pt}
\printbibliography
\balance
\end{document}